\documentclass[twocolumn, pra]{revtex4}
\usepackage{amsmath,amssymb,amsthm,epsfig,graphics,graphicx,hyperref}
\usepackage{verbatim}

\newtheorem{theorem}{Theorem}
\newtheorem{lemma}{Lemma}
\newtheorem{corollary}{Corollary}
\theoremstyle{definition}

\newtheorem*{remark}{Remark}
\newcommand{\bra}[1]{\langle #1|}
\newcommand{\ket}[1]{|#1\rangle}

\newcommand{\op}[2]{|#1\rangle \langle #2|}
\newcommand{\slocc}{\overset{\underset{\mathrm{SLOCC}}{}}{\longrightarrow}}

\begin{document}

\title{New Entanglement Monotones for W-type States}

\author{E. Chitambar}
\email{e.chitambar@utoronto.ca}
\author{W. Cui}
\email{cuiwei@physics.utoronto.ca}
\author{H.K. Lo}
\email{hklo@comm.utoronto.ca}
\thanks{\\ $^*{}^\dagger$ Authors contributed equally to this project.}

\affiliation{Center for Quantum Information and Quantum Control (CQIQC),
Department of Physics and Department of Electrical \& Computer Engineering,
University of Toronto, Toronto, Ontario, M5S 3G4, Canada}

\date{\today}

\begin{abstract}
In this article, we extend recent results concerning random-pair EPR distillation and the operational gap between separable operations (SEP) and local operations with classical communication (LOCC).  In particular, we consider the problem of obtaining bipartite maximal entanglement from an $N$-qubit W-class state (i.e. that of the form $\sqrt{x_0}\ket{00...0}+\sqrt{x_1}\ket{10...0}+...+\sqrt{x_n}\ket{00...1}$) when the target pairs are \textit{a priori} unspecified.  We show that when $x_0=0$, the optimal probabilities for SEP can be computed using semi-definite programming.  On the other hand, to bound the optimal probabilities achievable by LOCC, we introduce new entanglement monotones defined on the $N$-qubit W-class of states.  The LOCC monotones we construct can be increased by SEP, and in terms of transformation success probability, we are able to quantify a gap as large as 37\% between the two classes.  Additionally, we demonstrate transformations $\rho^{\otimes n}\to \sigma^{\otimes n}$ that are feasible by SEP for any $n$ but impossible by LOCC.
\end{abstract}

\maketitle

\section{Introduction}

Quantum entanglement is a celebrated aspect of quantum theory and represents one of the sharpest departures from the classical world.  From a practical perspective, entanglement provides a key tool for novel technologies such as quantum teleportation \cite{Bennett-1993a}, dense coding \cite{Bennett-1992b}, and entanglement-based quantum cryptography \cite{Ekert-1991a}.  Formally treating entanglement as a physical resource involves specifying a quantitative measure so that it makes sense to discuss ``how much'' entanglement a certain quantum system possesses.  For bipartite pure states, the von Neumann entropy serves as the unequivocal measure of entanglement \cite{Popescu-1997a}.  However, for multi-partite pure states and even mixed bipartite states, there does not appear to exist one unifying entanglement measure \cite{Linden-2005a, Acin-2003a}.  Instead, it seems more appropriate to quantify the amount of entanglement in a given system relative to some particular task or physical characteristic.  

A necessary (and arguably sufficient) property that every entanglement measure must satisfy is the so-called \textit{LOCC constraint} \cite{Vedral-1998a, Vidal-2000a, Horodecki-2000a, Horodecki-2001a, Plenio-2005a, Plenio-2007a}.  In a realistic multi-partite setting, each party will possess a laboratory in which he/she performs quantum measurements on only one part of the whole system.  At the same time, the parties may wish to coordinate their measurement strategies by using a classical communication channel to share their measurement outcomes.  This paradigm is known as LOCC (local operations and classical communication), and it describes the basic setting for nearly all practical quantum communication schemes.  The LOCC constraint means that entanglement cannot be increased on average by LOCC.  Therefore, a function $\mu$ fulfills the LOCC constraint if for any LOCC process that converts $\rho$ into $\sigma_i$ with probability $p_i$, the following inequality holds: $\mu(\rho)\geq\sum_ip_i\mu(\sigma_i)$.  

While it is very easy to describe the idea of LOCC operations, giving a precise mathematical description is notoriously difficult \cite{Bennett-1999a, Rains-1999a, Donald-2002a}.  For many purposes - such as upper bounding the success probability of some LOCC task - a finely-tuned description is not necessary.  Instead, one can turn to a more general (but not too general) class of quantum operations and see what's possible under this relaxation.  The most natural approximation to LOCC is the class of separable operations (SEP).  For an $N$-partite quantum system, an operation is called separable if it admits a Kraus operator representation $\mathcal{E}(\cdot)=\sum_\lambda A_\lambda(\cdot)A_\lambda^\dagger$ where $A_\lambda=M_{1,\lambda}\otimes  M_{2,\lambda}\otimes...\otimes M_{N,\lambda}$.  As every LOCC operation is built by a successive composition of local maps $\mathcal{E}^{(k)}\otimes\mathbb{I}^{(\overline{k})}$, it follows that every LOCC map is separable.  Compared to LOCC, the structure of SEP is easier to analyze, and studying it has been useful for proving LOCC impossibility results \cite{Rains-1997a, Kent-1998a, Vedral-1998a, Chefles-2004a, Horodecki-1998a, Stahlke-2011a}. 

A somewhat unexpected finding is the existence of separable operations that cannot be implemented by LOCC \cite{Bennett-1999a}.  A dramatic example of this is the phenomena of ``nonlocality without entanglement'' which refers to certain sets of product states that can be distinguished by SEP but not by LOCC \cite{Bennett-1999a, Bennett-1999b}.  Following the initial finding that LOCC $\subsetneq$ SEP, additional examples were constructed that demonstrated this fact \cite{DiVincenzo-2003a, Koashi-2007a, Duan-2007a, Chitambar-2009b, Cui-2011a}.  Like LOCC, separable operations have the property that they cannot generate entanglement.  Indeed, if a separable map is applied to a general separable state $\sum_ip_i\rho_i^{(1)}\otimes...\otimes\rho^{(N)}_i$, the resultant state will likewise be separable.  The fact that LOCC $\not=$ SEP then implies a certain irreversibility to the non-LOCC separable maps since these operations are unable to create entanglement, but nevertheless they require some pre-shared entanglement to be performed in the multi-partite setting.  Thus, such maps may be interpreted as the operational analog to ``bound entanglement'' \cite{Horodecki-1998a}, where the latter refers to multi-partite states that cannot be converted into pure entanglement but nevertheless require some initial entanglement to be created.  Consequently, studying the gap between LOCC and SEP is crucial to understanding the nature of quantum entanglement.

Unfortunately, very little quantitative research has been conducted into the difference between LOCC and SEP.  Thus it becomes difficult to say just how much more powerful SEP is than LOCC.  Previous numerical results that compared SEP versus LOCC for the task of distinguishing certain quantum states was very small in scale.  For instance, Ref. \cite{Bennett-1999a} demonstrated a minimum of $O(10^{-6})$ between the two classes (in terms of the attainable mutual information), while in Ref. \cite{Koashi-2007a}, optimal success probabilities in distinguishability were shown to diverge by at most $.8\%$.  Recently, however, we were able to provide the first appreciable gap between SEP and LOCC in terms of a 12.5\% difference in probability for successfully performing a particular state transformation \cite{Chitambar-2012a}.  In this article we vastly improve on our previous result and demonstrate a percent difference of 37\% between LOCC and SEP.  The key step in proving this result is the construction of new entanglement monotones for a particular subset of $N$-qubit states that can be increased by separable operations.

Specifically, we turn to the problem of randomly distilling an EPR pair from one copy of a multipartite W-class state, as first initiated by Fortescue and Lo \cite{Fortescue-2007a, Fortescue-2008a}.  An EPR random distillation refers to a transformation of multipartite entanglement into bipartite maximal pure entanglement in which the two parties sharing the final entanglement are allowed to vary among the different outcomes.  We denote such a transformation by \hypertarget{star}{}
\begin{align*} 
&\hspace{1.5cm}&&&&&\ket{\varphi}_{1...N}\to\big\{p_{ij},\ket{\Phi^{(ij)}}\big\}&\hspace{.8cm}&&&&&&(\star)
\end{align*}
where $\ket{\Phi^{(ij)}}$ is a maximally entangled two-qubit state shared between parties $i$ and $j$ obtained from $\ket{\varphi}_{1...N}$ with probability $p_{ij}$.  

It is often more convenient to represent transformation (\hyperlink{star}{$\star$}) by a \textit{distillation configuration graph} $\mathcal{G}=(V,E\subset V\times V)$ in which each party $i$ is assigned to a vertex $v_i\in V$, and an edge $(i,j)\in E$ is drawn between $v_i$ and $v_j$ if and only if $p_{ij}$ is nonzero (see Fig. \ref{RDfig1}).  Let $E_k\subset E$ denote the set of edges connected to vertex $v_k$.  

\begin{figure}[t]
\includegraphics[scale=0.6]{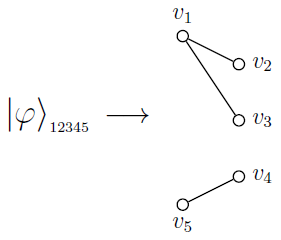}
\caption{\label{RDfig1} Graph representation $\mathcal{G}$ of one particular EPR random distillation configuration for the state $\ket{\varphi}_{12345}$.  Each edge represents a possible outcome EPR state shared between the two parties corresponding to the connected nodes.  The probability of obtaining a given edge is $p_{ij}>0$, and $E_k$ is the set of all edges connected to vertex $v_k$.} 
\end{figure} 

In terms of overall success probability, often the random distillation of some state can be more efficient than if entanglement is distilled to a fixed pair.  Perhaps the most impressive demonstration of this effect is the \textit{Fortescue-Lo Protocol} which performs transformation (\hyperlink{star}{$\star$}) on the three qubit W-state $\ket{W_3}=\sqrt{1/3}\left(\ket{100}+\ket{010}+\ket{001}\right)$ for any value of $p_{12}+p_{23}+p_{13}$ less than one; this should be compared to the maximum probability for transformation $\ket{W_3}\to\ket{\Phi^{(ij)}}$ that is $2/3$ for any $i\not=j\in\{1,2,3\}$ \cite{Fortescue-2007a}.  The Fortescue-Lo Protocol also extends to distilling an EPR pair from $\ket{W_N}$ with probability arbitrarily close to one.  Such a finding demonstrates the importance of considering random distillations in the multipartite setting.  

\subsubsection*{\upshape \bf Summary of Main Results and Article Outline}

This article compares the LOCC versus SEP feasible probabilities of transformation (\hyperlink{star}{$\star$}) when $\ket{\varphi}_{1...N}$ belongs to the $N$-qubit W-class of states, i.e. any state reversibly obtainable from $\ket{W_N}=\sqrt{1/N}\left(\ket{10...0}+\ket{01...0}+...+\ket{00...1}\right)$ by LOCC with a non-zero probability.  We begin our investigation in Section \ref{Sect:KT} with a review of results by Kinta\c{s} and Turgut on the subject of W-class transformations \cite{Kintas-2010a}.  There we also define the notation used throughout the paper.

In Section \ref{Sect:SEP} we show that for states of the form $\sqrt{x_1}\ket{10...0}+...+\sqrt{x_N}\ket{00...1}$, the possibility of transformation (\hyperlink{star}{$\star$}) by SEP can be phrased as a semi-definite programming feasibility question.  Thus, numerically it has an efficient solution.  When the initial state is $\ket{W_N}$, we are able to obtain simple necessary and sufficient criteria for transformation feasibility by studying the dual problem, as carried out in Appendix \ref{Apx:SDP}.  Note that the results of this section also provide LOCC upper bounds.

Next in Section \ref{Sect:monotone}, we turn to the LOCC setting specifically, and we introduce two new types of entanglement monotones defined on the $N$-qubit W-class of states.  To prove that these functions are monotonic under LOCC, we decompose a general LOCC transformation into a sequence of local weak measurements.  However, these functions are not monotonic under separable operations.  While a general separable measurement can also be decomposed into a sequence of weak measurements, these measurements need not be local, and our functions are sensitive precisely to this relaxation in constraint.  We prove that the monotones have operational meanings as the supremum success probabilities for the distillation of EPR states for certain distillation configuration graphs.  Moreover, these monotones can be saturated by an ``equal or vanish'' measurement scheme, which we further describe in Section \ref{Sect:monotone}.  Thus we are able to prove LOCC optimal rates for certain configuration graphs of transformation (\hyperlink{star}{$\star$}).  In particular, we solve the one-shot analog of ``entanglement combing'' studied by Yang and Eisert \cite{Yang-2009a} in which one particular party is selected to be a shareholder of the bipartite entanglement for each of the possible outcomes.  Formal comparisons between SEP and LOCC for transformation (\hyperlink{star}{$\star$}) are made in Section \ref{Sect:SEPvsLOCC}.

Finally, in Section \ref{Sect:ncopy} we move beyond the single-copy case and investigate a particular $n$-copy random distillation problem.  Interestingly, we we are able to show the existence of a state transformation $\op{\psi}{\psi}^{\otimes n}\to \rho^{\otimes n}$ that, for any $n$, is impossible by LOCC but always possible by SEP.  This result is the first of its kind.  Brief concluding remarks are then given in Section \ref{Sect:Conclusion}. 

\subsubsection*{\upshape \bf Relationship to Previous Work}
 
This article complements recent work we have conducted on the random distillation problem and its connection to the structure of LOCC \cite{Cui-2011a, Chitambar-2012a}.  In particular, Ref. \cite{Cui-2011a} presents a general LOCC procedure for completing transformation (\hyperlink{star}{$\star$}) on W-class states and also computes tight bounds for four qubit systems.  The distinguishing feature of this article is a solution to (\hyperlink{star}{$\star$}) by separable operations for a wide class of states and the construction of $N$-partite entanglement monotones that generalize those presented in \cite{Chitambar-2012a}.  Additionally, we consider here the many-copy variant of the random distillation problem, which has previously only been investigated in Ref. \cite{Fortescue-2008a}.  

\section{Notation and the Kinta\c{s} and Turgut Monotones} 
\label{Sect:KT}
Throughout the paper, we will be dealing exclusively with pure states $\ket{\varphi}_{1...N}$.  If ever we wish to express the state as the rank one density operator $\op{\varphi}{\varphi}_{1...N}$, we will denote it as $\varphi^{(1...N)}$.  For some operator $A$ acting on a multi-partite state space, we will let $A^{\Gamma_i}$ denote its partial transpose in the computational basis with respect to a party (or parties) $i$. 

It is often useful to consider two states equivalent if they can be reversibly converted from one to the other by LOCC with some nonzero probability.  Such a transformation is known as stochastic LOCC (SLOCC), and the well-known criterion for $\ket{\varphi}_{1...N}\slocc\ket{\varphi'}_{1...N}$ is the existence of invertible $M^{(k)}$ such that $\bigotimes_{k=1}^NM^{(k)}\ket{\varphi}_{1...N}=\ket{\varphi'}_{1..N}$ \cite{Dur-2000a}.  In this way, multipartite state space can then be partitioned into SLOCC equivalence classes.  

The $N$-party W-class is the set of states SLOCC equivalent to $\ket{W_N}=\sqrt{1/N}\left(\ket{10...0}+\ket{01...0}+\ket{00...1}\right)$, and such states take the form $\sqrt{x_0}\ket{00...0}+\sqrt{x_1}\ket{10...0}+...+\sqrt{x_n}\ket{00...1}$.  More importantly, even after a local unitary (LU) transformation - $\ket{0}\to\ket{0'}$ and $\ket{1}\to\ket{1'}$ - the component values $\sqrt{x_i}$ always remain unchanged for $N\geq 3$ \cite{Kintas-2010a}.  Therefore, we can uniquely characterize any W-class state by the $N$-component vector:
\begin{align}
\vec{x}=(x_1,&x_2,...,x_N)\notag\\
&\updownarrow\notag\\
\sqrt{x_0}\ket{00...0}+\sqrt{x_1}&\ket{10...0}+...+\sqrt{x_n}\ket{00...1},
\end{align} and $x_0=1-\sum_{i=1}^Nx_i$.  When $N=2$, uniqueness can be ensured by demanding that $x_0=0$ and $x_1\geq x_2$.

The order in value of these components will be highly important to our investigation.  Thus, we will often use the indices $\{n_1,n_2,...,n_N\}=\{1,2,...,N\}$ such that $x_{n_1}\geq x_{n_2}\geq ...\geq x_{n_N}$.  We let $n_1(\vec{x})$ denote the largest component in the state $\vec{x}=(x_1,...,x_N)$.

A main result of Kinta\c{s} and Turgut's work is proving that the component values, $-x_0$ and $x_i$ for $1\leq i\leq N$, are entanglement monotones \cite{Kintas-2010a}.  In other words, for an LOCC transformation converting $\vec{x}\to \vec{x}_\lambda$ with probability $p_\lambda$, the following relations hold:
\begin{align}
\label{Eq:KTmono}
x_0&\leq\sum_\lambda p_\lambda x_{\lambda,0}& x_i&\geq\sum_\lambda p_\lambda x_{\lambda,i}
\end{align}
for all $1\leq i\leq N$.
We will refer to these as the \textit{K-T monotones} and they place an upper bound of $\min\{x_i/y_i\}_{i=1...N}$ on the probability for any W-class transformation $\vec{x}\to\vec{y}$.  Recently, necessary and sufficient conditions were obtained for when this upper bound can be achieved \cite{Cui-2010b}.

To study the effects of measurement on a W-class state, first note that any measurement operator $A$ is a $2\times 2$ matrix expressible in the form $A=U\cdot \left(\begin{smallmatrix}\sqrt{a} & b  \\0 & \sqrt{c} \end{smallmatrix}\right )$ where $U$ is a unitary matrix.  Thus, up to a final local unitary operation, any local measurement corresponds to a set of upper triangular matrices $\{M_\lambda\}_\lambda$ with $\sum_\lambda M^\dagger_\lambda M_\lambda=\mathbb{I}$.  When it is party $k$ who performs the measurement, we will denote the measurement operators by $M_\lambda^{(k)}$.  It is easy to see that this measurement on state $\sqrt{x_0}\ket{00...0}+\sqrt{x_1}\ket{10...0}+...+\sqrt{x_N}\ket{00...1}$ will transform the components as:
\begin{align}
\label{starstar}
&\hspace{.6cm}&x_k&\to\frac{c_\lambda}{p_\lambda}x_k, &x_j&\to\frac{a_\lambda}{p_\lambda}x_j\;\;\;1\leq j\not=k\leq N,
\end{align}
where $p_\lambda$ is the probability that outcome $\lambda$ occurs.  We can simplify matters even further by noting that any transformation possible by LOCC can always be achieved by a protocol in which each party only performs two-outcome measurements \cite{Anderson-2008a}.  Since our chief concern is the possibility of transformations, we can assume without loss of generality that each local measurement consists of two upper triangular matrices $\{M^{(k)}_1, M^{(k)}_2\}$ whose entries are
\begin{align}
\label{Eq:constraints}
M^{(k)}_1&=\begin{pmatrix}\sqrt{a_1} & b_1  \\0 & \sqrt{c_1}\end{pmatrix} & M^{(k)}_2&=\begin{pmatrix}\sqrt{a_2} & b_2  \\0 & \sqrt{c_2}\end{pmatrix}
\end{align}
with $a_1+a_2=1$ and $c_1+c_2\leq 1$, in which equality is achieved by the latter if and only if $M^{(k)}_1$ and $M^{(k)}_2$ are both diagonal.

\section{Separable Transformations}
\label{Sect:SEP}

In this section we derive the conditions for which transformation (\hyperlink{star}{$\star$}) is possible by separable operations when the initial state is a W-class state with $x_0=0$.  As shown in the following lemma, the unique structure of such states allows for a major simplification in the analysis.

\begin{lemma}
\label{Lem:SEPsimp}
Suppose that $\{\Pi_\lambda:=M^{(1)}_\lambda\otimes...\otimes M^{(N)}_\lambda\}_{\lambda=1...t}$ corresponds to a complete measurement that achieves transformation {\upshape (\hyperlink{star}{$\star$})} with probabilities $p_{12},...,p_{N-1N}$ when $\ket{\varphi}_{1...N}=\sqrt{x_1}\ket{10...0}+...+\sqrt{x_N}\ket{00...1}$.  Then up to local unitary operations, there exists a measurement $\{\hat{M}^{(1)}_\lambda\otimes...\otimes \hat{M}^{(N)}_\lambda\}_{\lambda=1...2t}$ that achieves transformation {\upshape (\hyperlink{star}{$\star$})} with the same probabilities and with each $\hat{M}^{(k)}_\lambda$ being diagonal.
\end{lemma}
\begin{proof}
Up to an LU operation, each $M^{(k)}_\lambda$ takes the form
 $M^{(k)}_\lambda=\left(\begin{smallmatrix}\sqrt{a_{\lambda k}} & b_{\lambda k}  \\0 & \sqrt{c_{\lambda k}} \end{smallmatrix}\right)$ so that
\begin{equation}
\label{Eq:seppsimp}
\sum_\lambda\Pi_\lambda^\dagger\Pi_\lambda=\sum_\lambda\bigotimes_{k=1}^N \begin{pmatrix}a_{k\lambda}&\sqrt{a_{\lambda k}}b_{\lambda k}\\\sqrt{a_{\lambda k}}b^*_{\lambda k}&|b_{\lambda k}|^2+c_{\lambda k}\end{pmatrix}=\mathbb{I}.
\end{equation}
Let $\hat{M}^{(k)}_\lambda:=\left(\begin{smallmatrix}\sqrt{a_{\lambda k}}&0\\0&\sqrt{c_{\lambda k}}\end{smallmatrix}\right)$.  It is straightforward to see that the operators $\{\hat{\Pi}_{\lambda}:=\bigotimes_{k=1}^N\hat{M}^{(k)}_\lambda\}_{\lambda=1...t}$ correspond to an incomplete measurement that achieves transformation (\hyperlink{star}{$\star$}) with the same probabilities as the $\{\Pi_\lambda\}_{\lambda=1...t}$.  From Eq. \eqref{Eq:seppsimp}, the collection of separable operators $\big\{\bigotimes_{k=1}^N\left(\begin{smallmatrix}0&0\\0&|b_{\lambda k}|\end{smallmatrix}\right)\big\}_{\lambda=1...t}$ can be combined with $\{\hat{\Pi}_{\lambda}\}_{\lambda=1...t}$ to form a set which corresponds to a complete measurement.
\end{proof}

One immediate consequence of this lemma is that for any incomplete separable transformation of the form (\hyperlink{star}{$\star$}) with $\sum_\lambda\Pi_\lambda^\dagger\Pi_\lambda<\mathbb{I}$, we can always assume that $\mathbb{I}-\sum_\lambda\Pi_\lambda^\dagger\Pi_\lambda$ has a diagonal representation and is therefore separable.  As a result, when $\ket{\varphi}_{1...N}$ is a W-class state, it is sufficient to consider the feasible probabilities of transformation (\hyperlink{star}{$\star$}) under incomplete separable transformations.  

Now for measurement $\{\Pi_\lambda:=M^{(1)}_\lambda\otimes...\otimes M^{(N)}_\lambda\}_{\lambda=1...t}$, if we let $S_{ij}$ denote the set of all outcomes $\lambda$ such that $\Pi_\lambda\ket{\varphi}_{1...N}\propto\ket{\Psi^{(ij)}}$, we can form a Choi matrix $\Omega_{ij}$ for each edge $(i,j)\in E$ of the graph $\mathcal{G}$ \cite{Jamiolkowski-1972a}:
\begin{equation}
\Omega_{ij}=\sum_{\lambda\in S_{ij}}\Pi_\lambda\otimes \mathbb{I}\left(\bigotimes_{i=1}^N\Phi^{(ii')}\right)(\Pi^\dagger_\lambda)\otimes\mathbb{I}.
\end{equation}
Here, $\Pi_\lambda$ acts on systems $1,2,...,N$ while $\mathbb{I}$ is the identity acting on their copies $1',2',...,N'$.  By Lemma \ref{Lem:SEPsimp}, the $\Pi_\lambda$ can be taken as diagonal matrices so that $\Omega_{ij}$ has support only on the span of $\{\ket{i_1i_1}_{11'}\ket{i_2i_2}_{22'}...\ket{i_Ni_N}_{NN'}\}_{i_1,i_2,...i_N\in\{0,1\}}$. Furthermore, since all parties besides $i$ and $j$ hold pure states in the end, $M^{(k)}_\lambda$ must be a rank one matrix for $k\not=i,j$ and $\lambda\in S_{ij}$.  Thus, up to local unitaries and a permutation of spaces, $\Omega_{ij}$ has the form 
\[\Omega^{(ij)}=\chi^{(ii'jj')}\otimes\op{0}{0}^{(\overline{ii'jj'})}\]
where $\chi^{(ii'jj')}$ is effectively a separable $2\otimes 2$ density matrix having support on $\{\ket{mm}_{ii'}\ket{nn}_{jj'}\}_{m,n\in\{0,1\}}$; equivalently, $\chi^{(ii'jj')}$ has a positive partial transpose \cite{Horodecki-1996a}.
In terms of the Choi matrix, the condition of obtaining $\ket{\Phi^{(ij)}}$ with probability $p_{ij}$ is given by 
\begin{equation}
\label{Eq:Choitrans}
tr_{1'...N'}(\Omega_{ij}\varphi^{(1'...N')})=p_{ij}\Phi^{(ij)}\otimes\op{0}{0}^{(\overline{ij})}.
\end{equation}
Here we use the fact that $\varphi^{(1'...N')}$ is taken to have only real components.  Finally, the constraint that $\sum_{(i,j)\in E}\sum_{\lambda\in S_{ij}}\Pi^\dagger_\lambda\Pi_\lambda\leq\mathbb{I}$ is captured by
\begin{equation}
\label{Eq:incompleteChoi}
\sum_{(i,j)\in E}tr_{1...N}(\Omega_{ij})\leq\mathbb{I}.
\end{equation}
This construction is completely reversible such that given matrices satisfying the above conditions, we can always construct a separable measurement facilitating transformation (\hyperlink{star}{$\star$}) \cite{Cirac-2001a}.  Thus the necessary and sufficient conditions for a feasible separable map are $4\times 4$ complex matrices $\chi^{(ii'jj')}$ for all $(i,j)\in E$ which satisfy
\begin{align}
\label{Eq:ChoiPPT}
\chi^{(ii'jj')}&\geq 0\notag\\
[\chi^{(ii'jj')}]^{\Gamma_{i'j'}}&\geq 0,
\end{align}
as well as Equations \eqref{Eq:Choitrans} and \eqref{Eq:incompleteChoi}.  This is a semi-definite feasibility problem which can be efficiently solved using a variety of numerical tools \cite{Vandenberghe-1994a}.  Furthermore, duality theory can be used to analytically prove instances of infeasibility.  We perform such an analysis in Appendix \ref{Apx:SDP} for the initial state $\ket{\varphi_{1,...N}}=\ket{W_N}$.  The result is given by the following theorem, which also provides an LOCC upper bound.

\begin{theorem}
\label{Thm:Septhm1}
For $\ket{\varphi_{1...N}}=\ket{W_N}$, transformation {\upshape (\hyperlink{star}{$\star$})} with graph representation $\mathcal{G}$ is possible by separable operations if and only if
\begin{align}
\label{Eq:septhm}
\frac{N^2}{4}\sum_{(i,j)\in E}p_{ij}^2&\leq 1, & \frac{N}{2}\sum_{(i,j)\in E_k}p_{ij}&\leq 1,&&1\leq k\leq N.
\end{align}
\end{theorem}

\begin{remark}
\noindent In practice, it may be helpful to use the inequality $\sum_{i=1}^nx^2_i\geq\frac{1}{n}(\sum_{i=1}^tx_i)^2$ so that the first constraint in Eq. \eqref{Eq:septhm} becomes
\begin{equation}
\label{Eq:sepcons2}
\frac{N^2}{4|E|}\left(\sum_{(i,j)\in E}p_{ij}\right)^2\leq 1.
\end{equation}
\end{remark}

\section{Entanglement Monotones}
\label{Sect:monotone}

In this section, we introduce new entanglement monotones on the $N$-qubit W-class of states.  An important property of quantum measurements is the universality of weak measurements.  This means that any general measurement can be replaced by a sequence of measurements that obtains the same overall outcomes but only changes the state by an arbitrarily small increment with each individual measurement \cite{Bennett-1999a, Oreshkov-2005a}.  Consequently, to prove LOCC monotonicity of a given function, it is sufficient to prove it non-increasing on average under two-outcome infinitesimal measurements by a single party.  The full generality of this latter consideration was explored in Ref. \cite{Oreshkov-2006a}.  Here, a weak measurement of $\{M_1^{(k)}, M_2^{(k)}\}$ corresponds to $(a_1,c_1,a_2,c_2)$ lying in a small neighborhood of $(1/2,1/2,1/2,1/2)$, and the relatively simple structure of the W-class eases analysis in this infinitesimal setting. 

We define our monotones as follows.  For an $N$-party W-state $(x_1,x_2,...,x_N)$, set $\{n_1,n_2,...,n_N\}=\{1,2,...,N\}$ such that $x_{n_1}\geq x_{n_2}\geq ...\geq x_{n_N}$ and consider the continuous functions:
\begin{align}
\eta(\vec{x})&=x_{n_1}-\left(\frac{1}{x_{n_1}}\right)^{N-2}\prod_{i=2}^{N}(x_{n_1}-x_{n_i})\notag\\
\kappa(\vec{x})&=\sum_{i=2}^Nx_{n_i}+\eta(\vec{x}).
\end{align}
\begin{theorem}
\label{Thm:montonemain}
\begin{itemize}
\item[{}]
\item[{\upshape (I)}] $\eta$ is non-increasing on average for any single local measurement in which $n_{1}$ is
the same value for the initial and all possible final states,
\item[{\upshape (II)}] $\kappa$ is an entanglement monotone.  It is strictly decreasing on average for any non-trivial measurement by party $n_{1}$.
\end{itemize}
\end{theorem}
The three qubit form of this theorem has been proven in Ref. \cite{Chitambar-2012a}.  Here, in the general case, our proof technique will be very similar.
\begin{proof}
(I)
We consider case-by-case measurements of each party under the conditions of (I).  The function $\eta$ transforms as $\eta\to\eta_\lambda$ for $\lambda=1,2$, and we are interested in the average change: $\overline{\eta_\lambda}=p_1\eta_1+p_2\eta_2$ under infinitesimal measurements.  First suppose that party $n_1$ measures.  According to Eq. \eqref{starstar}, the average change in $\eta$ is
\begin{align}
\label{Eq:1}
\overline{\eta(\vec{x}_\lambda)}&=c_1x_{n_1}\left(1-\prod_{i=2}^N(1-\frac{a_1x_{n_i}}{c_1x_{n_1}})\right)\notag\\
&+c_2x_{n_1}\left(1-\prod_{i=2}^N(1-\frac{a_2x_{n_i}}{c_2x_{n_1}})\right).
\end{align}
We demonstrate that in the weak measurement setting, this quantity is maximized by equality of the upper bound: $c_1+c_2=1$. Indeed, we have
\begin{align}
\label{Eq:2}
\frac{\partial \overline{\eta_\lambda}}{\partial c_\lambda}|_{a_1=a_2=1/2 \atop c_1=c_2=1/2} &=x_{n_1}\bigg\{\left(1-\prod_{i=2}^N(1-\frac{x_{n_i}}{x_{n_1}})\right)\notag\\
&-\sum_{i=2}^N\frac{x_{n_i}}{x_{n_1}}\prod_{j\not=i}^N(1-\frac{x_{n_j}}{x_{n_1}})\bigg\}
\end{align}
for $\lambda=1,2$, and it suffices to show that this expression is strictly positive.  Now if we differentiate Eq. \eqref{Eq:2} with respect to any $x_{n_k}$ we obtain
\begin{align}
\prod_{i\not=k}^N(1-\frac{x_{n_i}}{x_{n_1}})-\prod_{i\not=k}^N(1-\frac{x_{n_i}}{x_{n_1}})+\sum_{i\not=k}\frac{x_{n_i}}{x_{n_1}}\prod_{j\not=i,k}^N(1-\frac{x_{n_j}}{x_{n_1}})\notag\\
=\sum_{i\not=k}\frac{x_{n_i}}{x_{n_1}}\prod_{j\not=i,k}^N(1-\frac{x_{n_i}}{x_{n_1}})\geq 0\qquad (2\leq k\leq N),\notag
\end{align}
and since Eq. \eqref{Eq:2} vanishes when $x_{n_k}=0$ for all $n_k$, it follows that for nonzero values of $x_{n_k}$, Eq. \eqref{Eq:2} is strictly positive.  Thus, the maximal change in $\eta$ occurs when $c_1+c_2=1$.  As we are interested in this upper bound, we will assume the measurement is characterized by $a\equiv a_1$, $1-a=a_2$, $c\equiv c_1$, and $1-c=c_2$.  We then have
\begin{align}
\label{Eq:3}
\eta-\overline{\eta(\vec{x}_\lambda)}&=-x_{n_1}\prod_{i=2}^N\left(1-\frac{x_{n_i}}{x_{n_1}}\right)+cx_{n_1}\prod_{i=2}^N\left(1-\frac{ax_{n_i}}{cx_{n_1}}\right)\notag\\
&+(1-c)x_{n_1}\prod_{i=2}^N\left(1-\frac{(1-a)x_{n_i}}{(1-c)x_{n_1}}\right).
\end{align}
Expanding this to second order about the point $(a,c)=(1/2,1/2)$ yields
\begin{align}
\eta-\overline{\eta(\vec{x}_\lambda)}&\approx 4(a-c)^2\sum_{i,j}\frac{x_{n_i}x_{n_j}}{x_{n_1}}\prod_{l\not=i,j}^N\left(1-\tfrac{x_{n_l}}{x_{n_1}}\right)\geq 0.
\end{align}
And this expression will be positive whenever $a\not=c$, which is whenever party $n_1$ performs a non-trivial measurement.  In the case in which party $n_i$ performs a measurement for some $i>1$, $\eta$ changes as  
\begin{align}
\overline{\eta(\vec{x}_\lambda)}=x_{n_1}&-(a_1x_{n_1}-c_1x_{n_i})\prod_{j\not=i}^N(1-\frac{x_{n_j}}{x_{n_1}})\notag\\&-(a_2x_{n_1}-c_2x_{n_i})\prod_{j\not=i}(1-\frac{x_{n_j}}{x_{n_1}})\leq \eta(\vec{x}).
\end{align}

(II)  We can always decompose a general transformation into a sequence of weak measurements for which each measurement either satisfies the conditions of (I), or its pre-measurement state $\vec{y}$ satisfies $y_{n_1}=y_{n_2}$.  In the first case, $\kappa$ is monotonic by part (I) and the fact that $\sum_{i=2}^Nx_{n_i}$ is non-increasing on average by the K-T monotones.  In the second case, we have $\kappa(\vec{y})=1-y_0$.  Since $1-y_{\lambda,0}$ is an upper bound on $\kappa(\vec{y}_\lambda)$ for each of the post-measurement states $\vec{y}_\lambda$, and $1-y_0$ is non-increasing on average by the K-T monotones, it follows that $\kappa(\vec{y})\geq\sum_\lambda p_\lambda\kappa(\vec{y}_\lambda)$.  Thus, $\kappa$ is an entanglement monotone in general.
\end{proof}

Theorem \ref{Thm:montonemain} also applies to any fixed collection of subsystems.  Indeed for $N$-qubit systems, let $S$ denote some subset of parties, and consider the unnormalized state $\vec{s}$ which has $|S|$ components, each belonging to a different party in $S$.  Then Theorem \ref{Thm:montonemain} also holds for the functions $\eta(\vec{s})$ and $\kappa(\vec{s})$.  The proof of this is exactly the same as above with the added note that whenever a measurement is performed by a party not in $S$, $\eta(\vec{s})$ and $\kappa(\vec{s})$ remain invariant on average, which follows from Eq. \eqref{starstar}.  

For example, in a 4-party system, let $S$ be parties 1, 2, and 3.  Now for any four-qubit state $\vec{x}$, take $\{x_{max},x_{mid},x_{min}\}=\{x_1,x_2,x_3\}$ such that $x_{max}\geq x_{mid}\geq x_{min}$.  Then, the function 
\begin{equation}
2x_{mid}+2x_{min}-\frac{x_{mid}x_{min}}{x_{max}}
\end{equation}
is an entanglement monotone.

The condition in part (I) of Theorem \ref{Thm:montonemain} can be extended beyond single measurements.  
\begin{corollary}
\label{Cor:eta}
Suppose the transformation $\vec{x}\to \{p_i, \vec{y_i}\}$ is possible by LOCC where $n_1(\vec{x})=n_1(\vec{y}_i)$ for all $i$.  Then $\eta(\vec{x})\geq \sum_ip_i\eta(\vec{y}_i).$
\end{corollary}
\begin{proof}
We can partition any transformation into sections where $n_1(\vec{x})$ is the largest component and where it is not.  By weak measurement theory, we can assume that when passing from one section to the other, we always first obtain a state $\vec{s}$ on the border such that $\eta(\vec{s})=s_{n_1(\vec{x})}$.  Therefore, since the $n_1(\vec{x})$ component is always monotonic by the K-T monotones \eqref{Eq:KTmono}, we have that $\eta$ will not have increased on average within any region for which $n_1(\vec{x})$ is not the largest component.  For sections when $n_1(\vec{x})$ is the largest, we know that $\eta$ is monotonic by part (I) of the previous theorem.
\end{proof}

\subsubsection*{\upshape \bf Interpretation of Monotones}

A natural question is whether the functions $\eta$ and $\kappa$ possess any physical interpretation.  Here we show that for states $\vec{x}$ having $x_0=0$, $2\eta(\vec{x})$ gives the optimal probability for transformation (\hyperlink{star}{$\star$}) when the configuration graph $\mathcal{G}$ consists of all edges connected to node $v_{n_1(\vec{x})}$.  We will refer to this as a ``combing transformation'' since it represents a single-copy version of the entanglement combing procedure described in Ref. \cite{Yang-2009a}.  On the other hand, $\kappa(\vec{x})$ gives the optimal probability when $\mathcal{G}$ is complete, i.e. each vertex is connected to every other one (see Fig. \ref{etakappafig}).  The following theorem gives a precise statement of this result.

\begin{theorem}
\label{Thm:RDprobs}
For an $N$-party W-state $\vec{x}=(x_1,x_2,...,x_N)$, let $P_{tot}$ be the optimal total probability of obtaining an EPR pair by LOCC, and $P_k$ the optimal total probability of party $k$ becoming EPR entangled.  Then
\begin{itemize}
\item[\upshape(I)] $P_{tot}<\kappa(\vec{x})$, and
\item[\upshape(II)] $P_k\leq\begin{cases}2 x_{k} \;\;\text{if}\;\; x_k<x_l\;\; \text{for some}\;\; $l$\\2\eta(\vec{x}) \;\;\text{if}\;\; x_k\geq x_l\;\; \text{for all}\;\; $l$.
\end{cases}$
\end{itemize}
When $x_0=0$, the upper bound in {\upshape (I)} can be approached arbitrarily close while in {\upshape (II)} it can be achieved exactly.
\end{theorem}

\begin{figure}[t]
\includegraphics[scale=0.6]{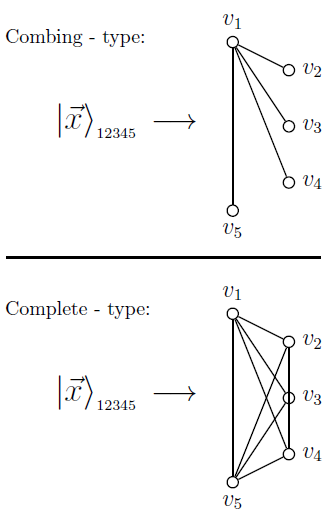}
\caption{\label{etakappafig} Distillation configurations for $\eta$ vs. $\kappa$.  TOP:  A ``combing-type'' distillation: when $x_0=0$, $2\eta(\vec{x})$ is the optimal probability for a random distillation in which party $n_1$ shares one half of each EPR pair.  BOTTOM:  A ``complete-type'' distillation: when $x_0=0$, $\kappa(\vec{x})$ gives the optimal probability for a random distillation in which the target pairs are any two of the parties.}
\end{figure}

\begin{proof}
First recall that $\kappa(\Phi^{(ij)})=1$.  Then the upper bounds follow from Theorem \ref{Thm:montonemain} and the K-T monotones.  Assume now that $x_0=0$.  To show that the upper bounds are effectively tight, we construct a specific protocol based on an ``equal or vanish'' (e/v) measuring scheme \cite{Fortescue-2007a}.  On its own, an e/v measuring scheme is just one way in which a W-class state $\ket{\varphi}_{1...N}$ can be converted into either EPR pairs or W states $\ket{W_m}$ for $3\leq m\leq N$.  Each party $k$ performs a two-outcome measurement for which outcome one is a state whose $k^{th}$ component equals the maximum component, and outcome two is a state whose $k^{th}$ component is zero.  The specific measurement operators are given by $M_1=diag[\sqrt{\tfrac{x_k}{x_{n_1}}},1]$ and $M_2=diag[\sqrt{1-\tfrac{x_k}{x_{n_1}}},0]$.  When each party does this, the possible resultant states are $\ket{\Phi^{(n_1k)}}$ for $n_2\leq k\leq N$, $\ket{W_m}$ for $3\leq m\leq N$, or a product state (see Fig. \ref{evfig}).  

For a complete-type distillation, the parties first perform e/v measurements and then implement the Fortescue-Lo Protocol on the resultant $\ket{W_m}$ states.  When $x_{n_1}=x_{n_2}$ for an initial state $\vec{x}$, a product state is never obtained by the e/v measurements, and the total success probability is therefore arbitrarily close to one.  When $x_{n_1}>x_{n_2}$, we prove the success rate by induction on the number of parties.  For $N=2$, the rate of $\kappa(\vec{x})=2x_{n_2}$ can be achieved \cite{Lo-1997a}.  Suppose now that probability $\kappa$ is obtained arbitrarily close with $N-1$ parties, and consider the $N$-party case.  If party $n_2$ is the first to perform an e/v measurement, then with probability $q$ this measurement will raise his component to equal the largest; i.e. the resultant state $\vec{y}$ has $y_{n_1}=y_{n_2}$.  Thus, random EPR distillation can be accomplished deterministically on $\vec{y}$.  For the ``vanish'' outcome occurring with probability $1-q$, the resultant state $\vec{z}$ is shared among $N-1$ qubits with $z_{n_i}=x_{n_i}\frac{x_{n_1}-x_{n_2}}{x_{n_1}(1-q)}$ for $n_i\not=2$.  By the inductive hypothesis, we then have:
\begin{align}
p_{tot}(\vec{x})&=q+(1-q)\big(1-\left(\frac{1}{z_{n_1}}\right)^{N-3}\prod_{i=3}^{N}(z_{n_1}-z_{n_i})\big)\notag\\
&=1-\frac{x_{n_1}-x_{n_2}}{x_{n_1}}\left(\frac{1}{x_{n_1}}\right)^{N-3}\prod_{i=3}^{N}(x_{n_1}-x_{n_i})\notag\\
&=1-\left(\frac{1}{x_{n_1}}\right)^{N-2}\prod_{i=2}^{N}(x_{n_1}-x_{n_i}).
\end{align}

For a combing-type distillation, when $x_k\leq x_l$ for some party $l$, $2x_k$ is known to be an achievable rate \cite{Turgut-2010a, Cui-2010a}.  When $x_k>x_l$ for all parties $l$, the procedure is for each party to perform an e/v measurement (in any order), except that when the first party $l$ obtains an ``equal'' outcome, a non-random EPR distillation is made between party $k$ and $l$.  This occurs with total probability $2x_l$, and a completely analogous inductive argument to the one given above shows that this full measurement scheme succeeds with probability exactly equal to $\eta(\vec{x})$.
\end{proof}

\begin{figure}[t]
\includegraphics[scale=0.6]{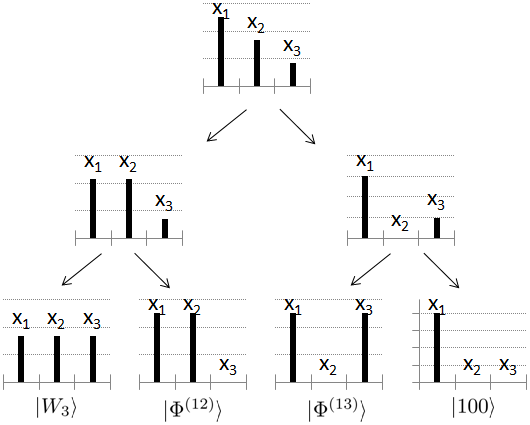}
\caption{\label{evfig} A three qubit ``equal or vanish'' measurement scheme.  The initial state is $(x_1,x_2,x_3)$ with $x_1>x_2>x_3$ and $x_0=0$.  Bob (party 2) measures first and either obtains a state in which his component is a maximum, or he becomes entangled from the other two.  In the next round Charlie (party 3) performs the same type of measurement.  The possible outcome states are $\ket{W_3}$, $\ket{\Phi^{(12)}}$, $\ket{\Phi^{(13)}}$, or a product state.  A ``complete-type'' distillation begins with this measurement scheme and then the Fortescue-Lo Protocol is performed on the $\ket{W_3}$ outcome.  A ``combing-type'' distillation is exactly this measurement scheme except that the pre-measurement state of $\ket{W_3}$ is converted into either $\ket{\Phi^{(12)}}$ or a product state (and not $\ket{W_3}$.}
\end{figure}

\begin{remark}
We make two remarks here.  First, for three qubit systems, combing and complete-type transformations represent the only two types of random distillations.  Thus, for three qubit states with $x_0=0$, Theorem \ref{Thm:RDprobs} gives a complete solution to transformation (\hyperlink{star}{$\star$}).  Second, a natural question is whether the ``equal or measurement'' scheme is always optimal for distilling EPR pairs.  In other words, for some random distillation configuration graph $\mathcal{G}$ is it always best to first perform e/v measurements, and then implement the Fortescue-Lo Protocol?  We have found that this is not the case and we describe specific counterexamples in Ref. \cite{Cui-2011a}.
\end{remark}

\section{SEP VS. LOCC}
\label{Sect:SEPvsLOCC}

In this section we use results from Section \ref{Sect:SEP} and Theorem \ref{Thm:RDprobs} to compare the distillation performances of SEP and LOCC.  In particular we consider an $N$-qubit combing-type distillation.  

The state we consider is $\ket{\psi_{1/2}}_{1...N}=\sqrt{\tfrac{1}{2}}\ket{10....0}+\sqrt{\tfrac{1}{2(1-N)}}\left(\ket{01...0}+...+\ket{00...1}\right)$.  By LOCC, the optimal probability for a combing-type transformation is
\begin{equation}
2\eta(\psi_{1/2})=1-(1-\frac{1}{N-1})^{N-1}\longrightarrow 1-e^{-1}
\end{equation}
where we have taken the limit for large $N$.  However, it is easy to see that the following separable operators (defined up to a reordering of spaces) represent a complete measurement which, with total probability one, will obtain an EPR pair shared by the first party:
\begin{align}
M_k&=\mathbb{I}^{(1)}\otimes\sqrt{\tfrac{1}{N-1}}\op{0}{0}^{(k)}+\op{1}{1}^{(k)}\bigotimes_{j\not=1,k}^N\op{0}{0}^{(j)}\notag\\
&\hspace{2cm}\text{for $1<k\leq N$},\notag\\
M_0&=\sqrt{\mathbb{I}-\sum_{i=1}^N M_k^\dagger M_k}.
\end{align}
We plot this separation between LOCC and SEP as a function of $N$ in Fig. \ref{SEPvsLOCC2}.

\begin{figure}[t]
\includegraphics[scale=0.6]{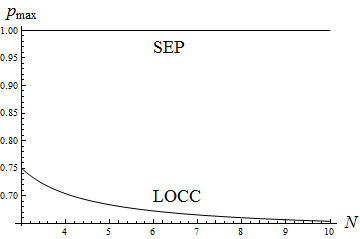}
\caption{\label{SEPvsLOCC2} LOCC vs. SEP for the maximum probability of party 1 become EPR entangled as a function of $N$ when initial state is $\sqrt{\tfrac{1}{2}}\ket{10....0}+\sqrt{\tfrac{1}{2(1-N)}}\left(\ket{01...0}+...+\ket{00...1}\right)$.  The LOCC probability is $1-(1-\frac{1}{N-1})^{N-1}$.  A gap of $37\%$ exists between SEP and LOCC.}
\end{figure}

\section{Multi-copy Distillations}
\label{Sect:ncopy}

So far we have only considered transformations of a single W-class state.  However, in this section, we consider a particular $n$-copy variant of transformation (\hyperlink{star}{$\star$}).  While the following discussion pertains to the tripartite case, its generalization to more parties is straightforward.  

Suppose the trio starts with $n$ copies of the state $\ket{\psi_{1/2}}=\sqrt{1/2}\ket{100}+1/2(\ket{010}+\ket{001})$, and they wish to distill $n$ EPR pairs such that Alice is always one of the shareholders (actually Bob and Charlie need not have the same components in the following argument).  The problem can be phrased as follows:
\begin{align}
\label{Eq:finitedistill}
&\qquad\ket{\psi_{1/2}}^{\otimes n}\to \ket{\Psi^{(AB)}}^{\otimes k}\ket{\Psi^{(AC)}}^{\otimes n-k}\notag\\
&\text{with probability $p_k=\tbinom{n}{k}/2^n$}\qquad \text{for $k=0,...n$}.
\end{align}
This is a combing-type transformation, and by the previous section we know that for any $n$, the transformation can always be completed with probability one by SEP.  On the other hand, even if the parties are allowed to act coherently on the $n$ copies of their local state, the following theorem still gives a no-go result.
\begin{theorem}
The transformation given by Eq. \eqref{Eq:finitedistill} is not possible by LOCC for any $n$.  Nor is it possible for any other distribution of the specified target states.
\end{theorem}
We give the proof below.  The only technical component needed is Lemma \ref{Lem2} which relies heavily on the special form of the state $\ket{\psi_{1/2}}$.  The main idea is that when viewed as a bipartite transformation with respect to A:BC, the reduced state entropies are the same for the initial and all the final states.  Consequently, the reduced state entropy must remain invariant for each measurement outcome in the LOCC protocol, and following the lines of Theorem 1 in Ref. \cite{Bennett-1999c}, this implies that Alice is restricted to only performing local unitaries.

However, due to the form of $\ket{\psi_{1/2}}$, invariance of the reduced state entropy also implies that Bob and Charlie can only perform local unitaries, as we will now show.  Without loss of generality, suppose that Bob acts first before Charlie in the protocol.  Since Alice can only have performed a local unitary up to this point, Bob and Charlie's reduced state is
\begin{equation}
(\frac{1}{2}\op{00}{00}+\frac{1}{2}\op{\Psi}{\Psi})^{\otimes n}=\frac{1}{2^n}\sum_{x\in\{0,1\}^{n}}\op{\tilde{x}}{\tilde{x}}
\end{equation}
where we introduce the notation that for a binary vector $x\in\{0,1\}^{n}$ with components $x_i\in\{0,1\}$, the corresponding string $\tilde{x}$ has symbolic components $\tilde{x}_i=00$ if $x_i=0$ and $\tilde{x}_i=\Psi$ if $x_i=1$. For example,
\[x=010\qquad\Rightarrow\qquad \ket{\tilde{x}}=\ket{00}\ket{\Psi}\ket{00}.\]
The reason for introducing this notation becomes evident in the following.
\begin{lemma}
\label{Lem2}
\noindent{\upshape (i)}For $x,y\in\{0,1\}^n$, let $S\subset\{0,1\}^n$ be the set such that $b\in S$ if $b_i=0$ whenever $x_i\cdot y_i\not=1$.  Then for any operator $A$ acting on Bob's system,
\begin{equation}
\label{Eq:Lem2}
\bra{\tilde{x}}A\otimes\mathbb{I}\ket{\tilde{y}}\propto\sum_{b\in S}\bra{x+b}A\ket{y+b}.
\end{equation}
\noindent\upshape (ii)
If $\bra{\tilde{x}}A\otimes\mathbb{I}\ket{\tilde{y}}=0$ for all $x\not=y\in\{0,1\}^n$, then $\bra{x}A\ket{y}=0$ for all $x\not=y\in\{0,1\}^n$. 
\noindent\upshape (iii)  If $\bra{\tilde{x}}A\otimes\mathbb{I}\ket{\tilde{x}}=k$ for all $x\in\{0,1\}^n$, then $\bra{x}A\ket{x}=k$ for all $x\in\{0,1\}^n$.
\end{lemma}
\begin{proof}
Part (i) can be verified from the relations $\bra{00}T\otimes\mathbb{I}\ket{00}\propto\bra{0}T\ket{0}$, $\bra{00}T\otimes\mathbb{I}\ket{\Psi}\propto\bra{0}T\ket{1}$, $\bra{\Psi}T\otimes\mathbb{I}\ket{00}\propto\bra{1}T\ket{0}$, and $\bra{\Psi}T\otimes\mathbb{I}\ket{\Psi}\propto\bra{1}T\ket{1}+\bra{0}T\ket{0}$.  For (ii), we use induction on $\log|S|$, i.e. on the number of coordinates simultaneously equal to $1$ in both $x$ and $y$.  By part (i), when $\log|S|=0$, then the statement is easily seen to be true from Eq. \eqref{Eq:Lem2} since the only $b\in S$ is the all zero vector $\vec{0}$.  Now suppose the claim is true when $\log|S|=m$, and consider two vectors $x$, $y$ such that $\log|S|=m+1$.  Again by part (i),
\[0=\bra{\tilde{x}}A\otimes\mathbb{I}\ket{\tilde{y}}\propto \sum_{\vec{0}\not=b\in S}\bra{x+b}A\ket{y+b}+\bra{x}A\ket{y}.\]
But for $\vec{0}\not=b\in S$, the strings $x+b$ and $y+b$ will have no more than $m$ coordinates that are both equal to $1$.  Therefore, by the inductive assumption each term in the sum vanishes, and so $\bra{x}A\ket{y}=0$.  Part (iii) can be proven by using a similar inductive argument and noting that for $\bra{\tilde{x}}A\otimes\mathbb{I}\ket{\tilde{x}}$, the proportionality factor in part (i) is $1/|S|$.
\end{proof}
\noindent Now, let $M$ be one Bob's measurement operators.  By invariance of the von Neumann entropy, we must have \cite{Nielsen-2000a}:
\begin{align}
n&=S\left(\frac{1}{2^np_M}\sum_{x\in\{0,1\}^{n}}M\otimes\mathbb{I}\op{\tilde{x}}{\tilde{x}}M^\dagger\otimes\mathbb{I}\right)\notag\\
&\leq H\bigg\{\frac{\bra{\tilde{x}}M^\dagger M\otimes\mathbb{I}\ket{\tilde{x}}}{2^np_M}\bigg\}_{x\in\{0,1\}^{n}}\leq n
\end{align}
which requires that $\bra{\tilde{x}}M^\dagger M\otimes\mathbb{I}\ket{\tilde{x}}$ is some positive constant for all $x\in\{0,1\}^{n}$ and the $M\otimes\mathbb{I}\ket{\tilde{x}}$ are orthogonal.  By (ii) and (iii) of Lemma \ref{Lem2}, this is only possible if $M^\dagger M$ is proportional to the identity.  In other words, $M$ is of the form $\sqrt{p}U$ for some unitary $U$.  

In the next round of measurement it will be Charlie's turn.  However, the above argument will apply for Charlie's measurement even after Bob performs an LU rotation.  Thus, in all rounds the parties can only perform local unitaries and therefore transformation \eqref{Eq:finitedistill} cannot be accomplished by any LOCC protocol.

Up to a conditional local unitary transformation, transformation \eqref{Eq:finitedistill} can be phrased as the mixed state transformation $\op{\psi_{1/2}}{\psi_{1/2}}^{\otimes n}\to\sigma^{\otimes n}$ where \[\sigma=1/2(\op{\Psi^{(AB)}}{\Psi^{(AB)}}\otimes\op{0}{0}+\op{\Psi^{(AC)}}{\Psi^{(AC)}}\otimes \op{1}{1}).\]
Here, the $\ket{0}$ and $\ket{1}$ is classical information accessible to all parties, and it encodes which particular duo holds the EPR state.  Thus, LOCC impossibility of transformation \eqref{Eq:finitedistill} means that the transformation $\op{\psi_{1/2}}{\psi_{1/2}}^{\otimes n}\to\sigma^{\otimes n}$ is LOCC infeasible.

Finally, we can consider the asymptotic setting and when the trio wishes to distill maximal entanglement with unit efficiency such that the entanglement is distributed equally to pairs Alice-Bob and Alice-Charlie.  More precisely, we seek for every $n$ an LOCC map $\Gamma^n$ such that 
\[tr[\Gamma^n(\psi^{\otimes n}_{1/2})\cdot\Psi^{(AB)\otimes n/2}\Psi^{(AC)\otimes n/2}]\to 1.\]
In fact, as given by the Entanglement Combing protocol of Ref. \cite{Yang-2009a}, this transformation is asymptotically feasible.  Moreover, their protocol holds for various distributions of final entanglement and not just equal shares between Alice-Bob and Alice-Charlie.  Consequently, we've shown that for particular state transformations, SEP $>$ LOCC regardless of the number of copies considered.  However, when the same transformations are considered in asymptotic form, we have that SEP $=$ LOCC.

\section{Conclusion}
\label{Sect:Conclusion}
In this article, we have studied the random distillation of W-class states by separable operations and LOCC.  Based on the transformation results of bipartite pure states \cite{Gheorghiu-2008a}, one may suspect that SEP and LOCC have equivalent transformation capabilities.  However, here we have shown that SEP is strictly more powerful.  

For separable operations, the general solution to transformation (\href{star}{$\star$}) can be solved by semi-definite programming optimization when $x_0=0$. This then places an upper bound on the problem for LOCC.  Tightening the LOCC bound requires analyzing each configuration graph in a case-by-case basis.  Two particular transformations we have considered are combing and complete-type transformations (Fig. \ref{etakappafig}).    Theorem \ref{Thm:RDprobs} provides an upper bound for the success probabilities of these transformations.  For states with $x_0=0$, the upper bounds can be approached arbitrarily close.  

To obtain these results, our general strategy has been to (i) start with a general W-class state and compute the combing or complete-type transformation probability using an ``equal or vanish'' protocol, and (ii) prove that the general probability expression (as a function of the components $x_i$) is an entanglement monotone.  This strategy isolates essential properties of LOCC beyond the tensor product structure of its measurement operators as it has generated entanglement monotones that can be increased under separable operations.

When $x_0\not=0$, we know these upper bounds are not tight, a prime example being the state $\sqrt{1-3s}\ket{000}+\sqrt{s}\left(\ket{100}+\ket{010}+\ket{001}\right)$ with $s>0$.  For a combing-type transformation of this state with Alice always being a shareholder in the outcome entanglement, the probability of success is upper bounded by $2\eta=2s$.  However, it is known that such a rate cannot be achieved \cite{Kintas-2010a, Cui-2010b}.  We leave it as an open problem to determine the optimal random distillation rates when $x_0\not=0$.

In terms of success probability, Fig. \ref{SEPvsLOCC2} shows a maximum percent difference of roughly 37\% for the combing-type distillation.  We conjecture that much larger gaps between SEP and LOCC exist than the ones shown in this article.  Even for the state $\ket{W_N}$, we predict that different distillation configuration graphs $\mathcal{G}$ restrict the feasible probabilities for LOCC much stronger than the separable upper bounds of Theorem \ref{Thm:RDprobs} (see Ref. \cite{Cui-2011a} for more details).  

Finally, we observe that for particular random distillations, the advantage of SEP over LOCC does not appear in the asymptotic setting, while it does when only finite resources are considered, regardless of the amount.  While we have shown this specifically for transformation \eqref{Eq:finitedistill}, the result holds true for more general combing-type transformations.  This suggests the intriguing conjecture that SEP and LOCC are operationally equivalent in the many-copy limit.  It is our hope that this article will lead to a deeper understanding of multipartite entanglement and the structure of LOCC.

\begin{acknowledgments}
We thank Andreas Winter, Debbie Leung, Graeme Smith, and John Smolin for providing helpful suggestions and discussing related material.  We also thank Jonathan Oppenheim, Ben Fortescue, Sandu Popescu and Runyao Duan for offering insightful comments on the subject.  We thank the financial support from funding agencies including NSERC, QuantumWorks, the CRC program and CIFAR. 
\end{acknowledgments}

\newpage

\appendix

\begin{widetext}

\section{Dual solution to $\ket{W_N}$ distillation by SEP}
\label{Apx:SDP}

We begin by writing Equations \eqref{Eq:incompleteChoi} and \eqref{Eq:ChoiPPT} in standard semi-definite programming (SDP) form.  Fix some encoding function $\phi:E\to|E|$ and define the matrices:
\begin{align}
F_1&=\left(\begin{smallmatrix}1&0&0&0\\0&0&0&0\\0&0&0&0\\0&0&0&0\end{smallmatrix}\right)\oplus\left(\begin{smallmatrix}1&0&0&0\\0&0&0&0\\0&0&0&0\\0&0&0&0\end{smallmatrix}\right)&
F_2&=\left(\begin{smallmatrix}0&1&0&0\\1&0&0&0\\0&0&0&0\\0&0&0&0\end{smallmatrix}\right)\oplus\left(\begin{smallmatrix}0&1&0&0\\1&0&0&0\\0&0&0&0\\0&0&0&0\end{smallmatrix}\right)&
F_3&=\left(\begin{smallmatrix}0&0&1&0\\0&0&0&0\\1&0&0&0\\0&0&0&0\end{smallmatrix}\right)\oplus\left(\begin{smallmatrix}0&0&1&0\\0&0&0&0\\1&0&0&0\\0&0&0&0\end{smallmatrix}\right)\notag\\
F_4&=\left(\begin{smallmatrix}0&0&0&1\\0&0&0&0\\0&0&0&0\\1&0&0&0\end{smallmatrix}\right)\oplus\left(\begin{smallmatrix}0&0&0&0\\0&0&1&0\\0&1&0&0\\0&0&0&0\end{smallmatrix}\right)&
F_5&=\left(\begin{smallmatrix}0&0&0&0\\0&0&0&1\\0&0&0&0\\0&1&0&0\end{smallmatrix}\right)\oplus\left(\begin{smallmatrix}0&0&0&0\\0&0&0&1\\0&0&0&0\\0&1&0&0\end{smallmatrix}\right)&
F_6&=\left(\begin{smallmatrix}0&0&0&0\\0&0&0&0\\0&0&0&1\\0&0&1&0\end{smallmatrix}\right)\oplus\left(\begin{smallmatrix}0&0&0&0\\0&0&0&0\\0&0&0&1\\0&0&1&0\end{smallmatrix}\right)\notag\\
F_7&=\left(\begin{smallmatrix}0&0&0&0\\0&0&0&0\\0&0&0&0\\0&0&0&1\end{smallmatrix}\right)\oplus\left(\begin{smallmatrix}0&0&0&0\\0&0&0&0\\0&0&0&0\\0&0&0&1\end{smallmatrix}\right)&
\end{align}
\begin{align}
\label{Eq:SDPGmat}
G^{(ij)}_1=&[-1]\;\;\;\bigoplus_{k=1}^{|E|} \;\;\;[0]\bigoplus_{k=1}^{\phi(i,j)-1}\;\;\;[0]_{4\times 4}\oplus F_1\bigoplus_{k=\phi(i,j)+1}^{|E|}[0]_{4\times 4}\notag\\
G^{(ij)}_2=&[0]\;\;\;\bigoplus_{k=1}^{|E|} \;\;\;[0]\bigoplus_{k=1}^{\phi(i,j)-1}\;\;\;[0]_{4\times 4}\oplus F_2\bigoplus_{k=\phi(i,j)+1}^{|E|}[0]_{4\times 4}\notag\\
G^{(ij)}_3=&[0]\;\;\;\bigoplus_{k=1}^{|E|} \;\;\;[0]\bigoplus_{k=1}^{\phi(i,j)-1}\;\;\;[0]_{4\times 4}\oplus F_3\bigoplus_{k=\phi(i,j)+1}^{|E|}[0]_{4\times 4}\notag\\
&...\notag\\
G^{(ij)}_7=&[0]\;\;\;\bigoplus_{k=1}^{\phi(i,j)-1}[0]\oplus [-1]\bigoplus_{k=\phi(i,j)+1}^{|E|}\;\;\;[0]\bigoplus_{k=1}^{\phi(i,j)-1}[0]_{4\times 4}\oplus F_7\bigoplus_{k=\phi(i,j)+1}^{|E|}[0]_{4\times 4}\notag\\
G_0=&[1]\;\;\;\bigoplus_{k=1}^{|E|}\;\;\;[1]\;\;\;\bigoplus_{\phi(i,j)=1}^{|E|}\;\;\;\left[\begin{pmatrix}0&0&0&0\\0&\tfrac{Np_{ij}}{2}&\tfrac{Np_{ij}}{2}&0\\0&\tfrac{Np_{ij}}{2}&\tfrac{Np_{ij}}{2}&0\\0&0&0&0\end{pmatrix}\oplus\begin{pmatrix}0&0&0&\tfrac{Np_{ij}}{2}\\0&\tfrac{Np_{ij}}{2}&0&0\\0&0&\tfrac{Np_{ij}}{2}&0\\\tfrac{Np_{ij}}{2}&0&0&0\end{pmatrix}\right].
\end{align}
Then Eqns. \eqref{Eq:ChoiPPT} and \eqref{Eq:incompleteChoi} are captured by the existence of $x_{k}^{(ij)}\in\mathbb{C}$ such that
\begin{equation}
\label{Eq:Semdef}
G_0+\sum_{(i,j)\in E}\sum_{m=1}^7x_m^{(ij)}G_m^{(ij)}\geq 0
\end{equation}
with the additional constraints that 
\begin{equation}
\label{Eq:probconst}
\sum_{(i,j)\in E_k}\frac{Np_{ij}}{2}\leq 1\;\;\;\text{for}\;\;\; 1\leq k\leq N.
\end{equation}
The dual problem to this asks 
\begin{align}
\label{Eq:SDPdual}
\max&\;\;-tr(ZG_0)\notag\\
\text{s.t}&\;\;0=tr(ZG_m^{(i,j)})\;\;\;\text{for all}\;\;\; G_m^{(i,j)}\notag\\
&Z\geq 0.
\end{align}
A critical relationship between the dual and primal formulations is that if \eqref{Eq:Semdef} can be satisfied for some $x_k^{(ij)}$, then for any $Z$ satisfying the constraints of $\eqref{Eq:SDPdual}$, we must have $tr(ZG_0)\geq 0$.  Thus infeasibility is proven by the existence of some $Z\geq 0$ such that $tr(ZG^{(ij)}_m)=0$ for all $G_m^{(i,j)}$ and $tr(ZG_0)< 0$.  We construct a certificate for infeasibility as follows.  For each $(i,j)\in E$, define the matrix:
\begin{align}
Z^{(ij)}=[\frac{1}{|E|}]\bigoplus_{k=1}^{\phi(i,j)-1}[0]\oplus[\frac{N^2p_{ij}^2}{4}]\bigoplus_{k=\phi(i,j)+1}^{|E|}[0]\bigoplus_{k=1}^{\phi(i,j)-1}[0]_{8\times 8}\oplus [0]_{4\times 4}\oplus\begin{pmatrix}1&0&0&\frac{-Np_{ij}}{2}\\0&0&0&0\\0&0&0&0\\\frac{-Np_{ij}}{2}&0&0&\frac{N^2p_{ij}^2}{4}\end{pmatrix}\bigoplus_{k=\phi(i,j)+1}^{|E|}[0]_{8\times 8}.
\end{align}
The claim is that the matrix 
\[Z:=\sum_{(i,j)\in E}Z^{(ij)}\]
is dual feasible with $tr(ZG_0)<0$ whenever $\frac{N^2}{4}\sum_{(i,j)\in E} p_{ij}^2>1$.  Indeed, it can easily be seen that $Z\geq 0$ and $tr[ZG_m^{(ij)}]=0$ for $1\leq m\leq 7$ and $(i,j)\in E$.  And finally,
\[tr[ZG_0]=1+\frac{N^2}{4}\sum_{(i,j)\in E}p_{ij}^2-\frac{N^2}{2}\sum_{(i,j)\in E}p_{ij}^2<0.\]
We have thus proven Theorem \ref{Thm:Septhm1}.
\end{widetext}




\bibliography{EricQuantumBib}

\end{document}